\documentclass[conference]{IEEEtran}

\usepackage{xcolor}

\usepackage{tcolorbox}

\usepackage{amsmath,amssymb,amsthm}
\newtheorem{lemma}{Lemma}
\newtheorem{theorem}{Theorem}
\newtheorem{corollary}{Corollary}

\usepackage{enumerate}

\usepackage{hyperref,color}
\hypersetup{
	colorlinks,
	linkcolor={blue!100!black},
	citecolor={blue!100!black},
	urlcolor={blue!80!black}
}

\title{Support Constrained Generator Matrices of\\Gabidulin Codes in Characteristic Zero}

\author{
\textbf{Hikmet~Yildiz}\textsuperscript{$*$},
~\textbf{Netanel~Raviv}$^\dagger$,
~and~\textbf{Babak~Hassibi}\textsuperscript{*}\\
\small\textsuperscript{*}Department of Electrical Engineering, California Institute of Technology, Pasadena CA 91125\\
$^\dagger$Department of Computer Science and Engineering, McKelvey School of
Engineering,\\
Washington University in Saint Louis, St. Louis, MO, 63130\\
\mbox{hyildiz@caltech.edu, netanel.raviv@wustl.edu, hassibi@caltech.edu}
}

\begin{document}
\maketitle

\begin{abstract}
    %Gabidulin codes can be seen as the rank-metric equivalent of Reed--Solomon codes, and their construction relies on finite field extensions. Since this construction fails over fields of characteristic zero, 
    Gabidulin codes over fields of characteristic zero were recently constructed by Augot~\textit{et al.}, whenever the Galois group of the underlying field extension is cyclic. In parallel, the interest in sparse generator matrices of Reed--Solomon and Gabidulin codes has increased lately, due to applications in distributed computations. In particular, a certain condition pertaining to the intersection of zero entries at different rows, was shown to be necessary and sufficient for the existence of the sparsest possible generator matrix of Gabidulin codes over finite fields. In this paper we complete the picture by showing that the same condition is also necessary and sufficient for Gabidulin codes over fields of characteristic zero. 
    
    Our proof builds upon and extends tools from the finite-field case, combines them with a variant of the Schwartz--Zippel lemma over automorphisms, and provides a simple randomized construction algorithm whose probability of success can be arbitrarily close to one. In addition, potential applications for low-rank matrix recovery are discussed.  
\end{abstract}

%%% To the best of my knowledge, keywords are not required for conference papers.
%\begin{IEEEkeywords}
%    keywords
%\end{IEEEkeywords}

\section{Introduction}
Over finite fields, Gabidulin codes~\cite{delsarte1978bilinear,gabidulin1985theory} can be seen as a rank-metric equivalent of Reed--Solomon codes, where instead of evaluating ordinary polynomials, one uses \textit{linearized polynomials} (i.e., whose only nonzero coefficients are for monomials whose degree is a nonnegative integer power of the field characteristic). To properly generalize this definition to fields of characteristic zero, it was recently suggested in~\cite{augot2018generalized} to employ $\theta$--polynomials, which are linear combinations of compositions of a generator~$\theta$ of the underlying Galois group of the field extension (that must be cyclic).

Independently, there has been a surge of interest lately in constructing sparsest generator matrices for Reed--Solomon and Gabidulin codes~\cite{yildiz2019gabidulin,dau2014existence,halbawi2014distributed,yildiz2019optimum,lovett2018mds}, for several applications in distributed computing. Since the rows of a generator matrix are codewords, each row cannot contain more than~$k-1$ zeros according to the Singleton bound, where~$k$ is the dimension of the code. The so called GM--MDS conjecture, posed by~\cite{dau2014existence} and solved by~\cite{yildiz2019optimum} and~\cite{lovett2018mds}, asserts that this maximum number of zeros at every row is attainable, as long as a certain condition regarding the position of zeros is satisfied. Specifically, this condition requires the zero-entries at every set of rows to intersect in at most~$k$ minus the number of rows in the intersection.
%for a code of length~$n$ and dimension~$k$, let~$\mathcal Z_i\subseteq\{1,2,\ldots,n\}$ be the positions of zeros in the~$i$'th row of a generator matrix, for~$i\in[k]$. Then, a necessary and sufficient condition for the existence of a generator matrix~$\mathbf G$ with~$\mathbf G_{i,j}=0$ for every~$i$ and~$j$ such that~$j\in\mathcal Z_i$, is that 
%\begin{align*}\textstyle
%\left|\bigcap_{i\in\Omega}\mathcal Z_i\right|+|\Omega|\leq k,
%\quad\forall\varnothing\neq\Omega\subset[k].
%\end{align*}

In this paper we complete the picture by showing that the same condition is necessary and sufficient for the existence of sparse generator matrices for Gabidulin codes over fields of characteristic zero. We note that while the proof of the equivalent condition for Reed--Solomon codes is identical for finite fields and fields of characteristic zero, for Gabidulin codes this is \textit{not} the case, and the proof from~\cite{yildiz2019gabidulin} fails over the latter fields. However, by adopting notions from the Reed--Solomon equivalent (the ``Simplified GM--MDS conjecture'' \cite[Thm.~3]{yildiz2019optimum}), and combining with a variant of the well-known Schwartz--Zippel lemma, we are able to resolve the problem over fields of characteristic zero. Moreover, our proof also provides a randomized construction algorithm whose probability of success can be arbitrarily high; similar randomized construction algorithms exist for the finite variants of the problem, but their probability of success is lower.

Beyond their application in network coding~\cite{silva2008rank}, space-time codes~\cite{lusina2003maximum}, and cryptography~\cite{gabidulin1991ideals}, Gabidulin codes have applications in \textit{low rank matrix recovery}~\cite{muelich2017low} (LRMR), which is normally performed over fields of characteristic zero. In this problem, one reconstructs a low-rank matrix from a given set of linear measurements. If these linear measurements are given by multiplication of the unknown matrix by a parity-check matrix of a Gabidulin code, this problem reduces to syndrome decoding of the respective zero codeword. Since the parity-check matrix of a Gabidulin code has a similar structure to that of the generator matrix~\cite[Prop.~8]{augot2018generalized}, our results imply that when performing LRMR with Gabidulin codes, one may employ linear measurements that depend on a small number of entries of the unknown matrix.

The problem is formally stated in Section~\ref{section:ProblemSetup}, alongside necessary mathematical background. Our main results are summarized in Section~\ref{section:mainResults}, and proved in Section~\ref{sec:proofs} by using auxiliary claims given in Section~\ref{section:MoreOn}. 

\subsection{Notations}
Let $[n]=\{1,2,\dots,n\}$.
Denote the dimension of a subspace~$V$ over a field $\mathsf F$  by $\dim_{\mathsf F} V$ and the span of the elements in a set~$S$ over the field $\mathsf F$ by $\operatorname{span}_{\mathsf F} S$. The (total) degree of a (multivariate) polynomial~$f$ is denoted by~$\deg f$ (e.g. $\deg(x^2y^2+x^3)=4$). % represents the total degree of a multivariate polynomial $f$.
For an $m\times n$ matrix $\mathbf X$ and $I\subseteq[m],J\subseteq[n]$, $\mathbf X_{I,J}$ is the submatrix with the rows and columns indexed in $I$ and $J$ respectively.
Let $\mathbf X_{I,:}=\mathbf X_{I,[n]}$ and $\mathbf X_{:,J}=\mathbf X_{[m],J}$ and when $I$ or $J$ has a single element, we sometimes write the element only instead of the set.

\section{Problem Setup}\label{section:ProblemSetup}
In this section we will first provide a brief background on cyclic Galois extensions.
Then, we will define rank metric codes and Gabidulin codes.
Finally, we will define our problem, namely, finding Gabidulin codes with support constrained generator matrices over a field of characteristic zero.

\subsection{Field extensions}
Let $\mathsf E/\mathsf F$ be a field extension of finite degree, i.e. the dimension of $\mathsf E$ as a vector space over $\mathsf F$ is finite, and let $\dim_{\mathsf F}\mathsf E=m$.
The automorphism group of $\mathsf E/\mathsf F$, $\operatorname{Aut}(\mathsf E/\mathsf F)$, is the set of automorphisms of $\mathsf E$ that fix $\mathsf F$, i.e.
\begin{align*}
    \operatorname{Aut}(\mathsf E/\mathsf F)=\{
    &\theta : \mathsf E\to\mathsf E \text{ automorphism}\mid
    \forall x\in\mathsf F, \theta(x)=x\},
\end{align*}
with the group operation of function composition~$\circ$.
If $\left|\operatorname{Aut}(\mathsf E/\mathsf F)\right|=m$, $\mathsf E/\mathsf F$ is called a Galois extension, in which case, $\operatorname{Aut}(\mathsf E/\mathsf F)$ is also denoted by $\operatorname{Gal}(\mathsf E/\mathsf F)$ and is called the Galois group of $\mathsf E/\mathsf F$.

In this paper, we will focus on cyclic Galois extensions, whose Galois group is a cyclic group of order $m$:
\begin{equation*}
    \operatorname{Gal}(\mathsf E/\mathsf F)= \{\theta^0,\theta^1,\dots,\theta^{m-1}\}
\end{equation*}
where the automorphism $\theta$ is the generator and $\theta^{i+1}=\theta\circ\theta^i$ for every $i\geq 0$. Notice that $\theta^m=\theta^0$ is the identity automorphism.

For example, for finite fields, when $\mathsf F=\mathbb F_q$ and $\mathsf E=\mathbb F_{q^m}$, the Galois group is cyclic of order $m$ with the generator automorphism $\theta(x)=x^q$:
\begin{align*}
    \operatorname{Gal}(\mathbb F_{q^m}/\mathbb F_q)
    =\left\{x, x^q, x^{q^2},\dots, x^{q^{m-1}}\right\}.
\end{align*}
For infinite fields, when $\mathsf F=\mathbb Q$ is the set of rational numbers and $\mathsf E=\mathbb Q(\zeta_n)$, where $\zeta_n$ is the $n$'th root of unity,
$\mathbb Q(\zeta_n)/\mathbb Q$ is a Galois extension of degree $\varphi(n)$,
where $\varphi(n)$ is the Euler's phi function
($\mathbb Q(\zeta_n)$ is called the $n$'th cyclotomic field and an interested reader can refer to~\cite{marcus1977number}).
Its Galois group is isomorphic to the multiplicative group $\mathbb Z^*_n$ of integers modulo~$n$.
Since~$\mathbb Z_n^*$ is cyclic for $n=p^a,2p^a$~\cite{Mathworld}, where~$p$ is any odd prime and~$a$ is any positive integer, it follows that for these values of~$n$ we have that~$\mathbb{Q}(\zeta_n)$ is a cyclic Galois extension of degree $m=\varphi(n)=p^{a-1}(p-1)$.
%If $n=p^a$ or $2p^a$ for an odd prime $p$, it is a cyclic Galois extension of degree $m=\varphi(n)=p^{a-1}(p-1)$.
It is also possible to define cyclic extensions of $\mathbb Q$ for any degree~$m$ by considering subfields of $\mathbb Q(\zeta_p)$ for an odd prime $p$ such that $p-1$ is divisible by $m$.

\subsection{Rank metric codes}
A linear rank metric code, $[n,k,d]_{\mathsf E/\mathsf F}$, over a field extension $\mathsf E/\mathsf F$ is an $\mathsf E$--subspace $\mathcal C$ of $\mathsf E^n$ of dimension $k$ with the rank distance
\begin{align}\label{eq:rankdistance}
    d=d_R(\mathcal C)\triangleq \min_{0\neq \mathbf c\in\mathcal C}\dim_{\mathsf F}(\operatorname{span}_{\mathsf F}\{c_1,\dots,c_n\})
\end{align}
where $c_1,\dots,c_n\in\mathsf E$ represent the entries of $\mathbf c\in\mathsf E^n$.
By fixing an ordered basis of~$\mathsf E$ over~$\mathsf F$, the elements of $\mathsf E$ can be considered as vectors in $\mathsf F^m$, and then
the codewords (i.e. the elements of $\mathcal C\subset\mathsf E^n$) can be viewed as $m\times n$ matrices over $\mathsf F$. Then, this definition of the rank distance in (\ref{eq:rankdistance}) is equivalent to the minimum of the rank of the matrix representation of a nonzero codeword.

Notice that by definition in (\ref{eq:rankdistance}), the rank distance of $\mathcal C$ can be upper bounded by the Hamming distance,
$d_H(\mathcal C)\triangleq \min_{0\neq c\in\mathcal C}\|c\|_0$,
where $\|c\|_0$ is the number of nonzero entries of $c$.
Therefore, the Singleton bound can be written for the rank distance as well:
\begin{align}\label{equation:dRledH}
d_R(\mathcal C)\leq d_H(\mathcal C)\leq n-k+1.
\end{align}
The codes with $d_R(\mathcal C)=n-k+1$ are called maximum rank distance (MRD), for which we write $[n,k]_{\mathsf E/\mathsf F}$ by omitting $d$.
A generator matrix for an $[n,k,d]_{\mathsf E/\mathsf F}$ code $\mathcal C$ is a $k\times n$ matrix over $\mathsf E$ whose rows form a basis for $\mathcal C$.

\subsection{Gabidulin codes}
Gabidulin codes are defined as the row space of the $k\times n$ matrix
\begin{align}\label{eq:gabidulin}
    \begin{bmatrix}
        \theta^0(x_1) & \theta^0(x_2) & \cdots & \theta^0(x_n)\\
        \theta^1(x_1) & \theta^1(x_2) & \cdots & \theta^1(x_n)\\
        \vdots & \vdots & & \vdots\\
        \theta^{k-1}(x_1) & \theta^{k-1}(x_2) & \cdots & \theta^{k-1}(x_n)
    \end{bmatrix}\in\mathsf E^{k\times n}
\end{align}
where $\theta\in\operatorname{Aut}(\mathsf E/\mathsf F)$ and $x_1,\dots,x_n\in\mathsf E$ are $\mathsf F$--linearly independent (notice that this requires $n\leq m=\dim_{\mathsf F}\mathsf E$).
Note that Gabidulin codes can be seen as evaluation codes of the so-called~$\theta$--polynomials; a $\theta$--polynomial is a function~$f:\mathsf E\to\mathsf E$ of the form~$f(x)=\sum_i f_i \theta^i(x)$ for~$f_i\in\mathsf{E}$, and every codeword in a Gabidulin code is the evaluations of some $\theta$--polynomial of~$\theta$--degree at most~$k-1$.
Note also that the generator matrix can be chosen as the product of any $k\times k$ invertible matrix over $\mathsf E$ and the matrix in (\ref{eq:gabidulin}).

Originally, this was defined by Delsarte \cite{delsarte1978bilinear} and Gabidulin \cite{gabidulin1985theory} for the finite fields, when $\mathsf F=\mathbb F_q$, $\mathsf E=\mathbb F_{q^m}$, and $\theta(x)=x^q$, as the first general constructions of MRD codes over finite fields.
Later~\cite{augot2018generalized}, it was extended to fields of characteristic zero and it was shown that when $\mathsf E/\mathsf F$ is a cyclic Galois extension and~$\theta$ is the generator of $\operatorname{Gal}(\mathsf E/\mathsf F)$, this extension of Gabidulin codes also gives an $[n,k]_{\mathsf E/\mathsf F}$ MRD code \cite{augot2018generalized}.
In the rest of the paper, we will assume that $\mathsf E/\mathsf F$ is a cyclic Galois extension of order~$m$ and~$\mathsf{F}$ is of characteristic zero.

\subsection{Problem definition}
We consider the problem of finding an $[n,k]_{\mathsf E/\mathsf F}$ MRD code whose generator matrix $\mathbf G\in\mathsf E^{k\times n}$ has support constraints. 
We describe the support constraints through the subsets $\mathcal Z_1,\mathcal Z_2,\dots,\mathcal Z_k\subset[n]$ as
\begin{align}\label{eq:zeroconstraints}
    \mathbf G_{ij} = 0,\qquad \forall j\in\mathcal Z_i, i=1,2,\dots,k.
\end{align}

Over finite fields, this problem was studied in \cite{yildiz2019gabidulin} and it was shown that a necessary and sufficient condition for the existence of MRD codes under support constraints described by the $\mathcal Z_i$ is
\begin{align}\label{eq:ineqcond}\textstyle
    \left|\bigcap_{i\in\Omega}\mathcal Z_i\right|+|\Omega|\leq k,
    \quad\forall\varnothing\neq\Omega\subseteq[k].
\end{align}
The same condition also appears in the GM--MDS conjecture for MDS codes (i.e. $d_H=n-k+1$, see~\cite{dau2014existence}, and also~\cite{yan2013algorithms,halbawi2014distributed}) which was proven in \cite{yildiz2019optimum} and \cite{lovett2018mds}.

Over infinite fields, the fact that \eqref{eq:ineqcond} is necessary can be shown similar to~\cite{yildiz2019optimum}, since MRD codes are also MDS~\eqref{equation:dRledH}, and since the proof in~\cite{yildiz2019optimum} applies to both finite and infinite fields.
However, a similar proof to \cite{yildiz2019gabidulin} cannot be applied to show that (\ref{eq:ineqcond}) is sufficient when $\mathsf F$ has characteristic zero.
The reason is that in finite fields, since the generator matrix in~\eqref{eq:gabidulin} consists of entries in the form of polynomials in the $x_i$'s, which, in one step of the proof, allows to reduce the problem to a similar one with a smaller parameter, whereas in the characteristic zero, the entries are in the form of $\theta$--polynomials (defined in \cite{augot2018generalized}) and applying the same step turns the problem into one of a different kind.
Hence, in this paper,
we will show that (\ref{eq:ineqcond}) is sufficient for the existence of $[n,k]_{\mathsf E/\mathsf F}$ MRD codes under the support constraints on the generator matrix given in (\ref{eq:zeroconstraints}) when~$\mathsf F$ has characteristic zero.

\section{Main Results}\label{section:mainResults}
In this section, we present our main results on the existence of MRD codes in characteristic zero (see Theorem~\ref{thm:main}) and the best achievable rank distance for the cases where there does not exist any (see Corollary \ref{corollary}). Also, we will give a randomized algorithm for the code construction. The proofs of the theorems will be given in Section \ref{sec:proofs}.

\begin{theorem}\label{thm:main}
    Let $\mathsf E/\mathsf F$ be a cyclic Galois extension of degree $m$ such that $\mathsf F$ has characteristic zero.
    For some $k\leq n\leq m$, let $\mathcal Z_1,\dots,\mathcal Z_k\subset[n]$ satisfy (\ref{eq:ineqcond}).
    Then, there exists an $[n,k]_{\mathsf E/\mathsf F}$ Gabidulin code with a generator matrix satisfying the constraints in (\ref{eq:zeroconstraints}).
\end{theorem}

If the $\mathcal Z_i$ do not satisfy (\ref{eq:ineqcond}), then as given in \cite{yildiz2019gabidulin} and \cite{yildiz2019optimum}, $d_R\leq d_H\leq n+1-\max\limits_{\varnothing\neq\Omega\subseteq[k]}\left(\left|\bigcap_{i\in\Omega}\mathcal Z_i\right|+|\Omega|\right) < n-k+1$ and hence, an MRD code does not exist.
For this case, Corollary~\ref{corollary} below (which is the analog of \cite[Thm.~2]{yildiz2019gabidulin}) shows that this upper bound is achievable by the subcodes (i.e., the subspaces) of Gabidulin codes.

\begin{corollary}\label{corollary}
    In Theorem~\ref{thm:main}, if the $\mathcal Z_i$ do not satisfy (\ref{eq:ineqcond}), then there exists an $[n,k,n-\ell+1]_{\mathsf E/\mathsf F}$ subcode of an $[n,\ell]_{\mathsf E/\mathsf F}$ Gabidulin code, which satisfies (\ref{eq:zeroconstraints}), where
    \begin{align}\textstyle
        \ell = \max\limits_{\varnothing\neq\Omega\subseteq[k]}
        \left(\left|\bigcap_{i\in\Omega}\mathcal Z_i\right|+|\Omega|\right)
    \end{align}
\end{corollary}
\begin{proof}
    Define $\mathcal Z_{k+1}=\cdots=\mathcal Z_{\ell}=\varnothing$. Then, for any nonempty $\Omega\subseteq[\ell]$, we have that
    $\left|\bigcap_{i\in\Omega}\mathcal Z_i\right| + |\Omega|\leq\ell
    $. Hence, by Theorem~\ref{thm:main}, there exists an $[n,\ell,n-\ell+1]_{\mathsf E/\mathsf F}$ Gabidulin code with an $\ell\times n$ generator matrix $\mathbf G$ having zeros dictated by $\mathcal Z_1,\dots, \mathcal Z_{\ell}$. The first $k$ rows of $\mathbf G$ will generate a subcode whose rank distance $d_R$ is as good as the Gabidulin code:  $d_R\geq n-\ell+1$. Furthermore, $n-\ell+1$ is an upper bound on $d_H$ \cite{yildiz2019optimum}.
    Therefore, $n-\ell+1\leq d_R\leq d_H\leq n-\ell+1$.
    Hence, $d_R=n-\ell+1$.
\end{proof}

\subsection{Code Construction}
Fix an $\mathsf F$--basis $\{b_1,\dots,b_m\}$ for $\mathsf E$ and 
assume that the conditions for the $\mathcal Z_i$ in Theorem~\ref{thm:main} are satisfied, i.e. $\mathcal Z_1,\dots,\mathcal Z_k\subset[n]$ satisfy (\ref{eq:ineqcond}).
Then, each $\mathcal Z_i$ has at most $k-1$ elements by applying (\ref{eq:ineqcond}) with $|\Omega|=1$.
In \cite[Thm.~2]{dau2014existence} and \cite[Corollary~3]{yildiz2019gabidulin}, it is shown that one can keep adding elements to these sets from $[n]$ without violating any of the inequalities in (\ref{eq:ineqcond}) until each $\mathcal Z_i$ has exactly $k-1$ elements.
Note that adding elements to these sets will only put more zero constraints on the generator matrix.
Therefore, without loss of generality, we can assume that $|\mathcal Z_i|=k-1$ for all $i$ along with (\ref{eq:ineqcond}).
Then, we construct a generator matrix for a rank metric code in a randomized manner as described below:

\begin{tcolorbox}[colframe=black,colback=white, sharp corners,colbacktitle=white,coltitle=black,boxrule=0.45pt]
    \textbf{Inputs:} A finite nonempty set $S\subset\mathsf F$ and  subsets $\mathcal Z_1,\dots,\mathcal Z_k\subset[n]$ satisfying (\ref{eq:ineqcond}).
    
    \textbf{Steps:}
    \begin{itemize}
        \item Add elements to the $\mathcal Z_i$'s from $[n]$ (if necessary) by following the algorithm given in \cite[Thm.~2]{dau2014existence} so that they all have \emph{exactly} $k-1$ elements and \emph{still} satisfy (\ref{eq:ineqcond}).
        \item Choose $(\gamma_{ij})_{i\in[n],j\in[m]}$ uniformly at random from $S$.
        \item Let $x_i=\sum_{j=1}^m\gamma_{ij}b_j$ for $i\in[n]$.
        \item Construct $\mathbf A\in\mathsf E^{k\times n}$ as in (\ref{eq:gabidulin}) in terms of $x_1,\dots,x_n$.
        \item Define $\mathbf T\in\mathsf E^{k\times k}$ as
        \begin{align}\label{eq:matrixt}
            \mathbf T_{ij} = 
            \det\begin{bmatrix}\mathbf e_j&\mathbf A_{:,\mathcal Z_i}\end{bmatrix},\quad i,j\in[k]
        \end{align}
        where $\mathbf e_j$ is the column vector with $1$ at the $j$th entry and $0$'s elsewhere
        (Note that $|\mathcal Z_i|=k-1$).
    \end{itemize}

    \textbf{Output:} The generator matrix $\mathbf G=\mathbf T\cdot\mathbf A\in\mathsf E^{k\times n}$.
\end{tcolorbox}

By Lemma \ref{lemma:general} below, $\mathbf G$ in the above construction is guaranteed to satisfy (\ref{eq:zeroconstraints}) for any inputs.
\begin{lemma}\label{lemma:general}
    Let $\mathcal Z_1,\dots,\mathcal Z_k\subset[n]$ be subsets of size $k-1$.
    For a given $k\times n$ matrix $\mathbf A$, a $k\times k$ matrix $\mathbf T$ (over the same field as $\mathbf A$) satisfying $(\mathbf T\cdot\mathbf A)_{ij}=0$ for every $j\in\mathcal Z_i$ and $i\in[k]$ can be given as
    in (\ref{eq:matrixt}).
\end{lemma}
\begin{proof}
    For a fixed~$i\in[k]$, the statement $(\mathbf T\cdot \mathbf A)_{ij}=0$ for every $j\in\mathcal Z_i$ is equivalent to the equation $\mathbf T_{i,:}\cdot\mathbf A_{:,\mathcal Z_i}=0$.
    A solution $\mathbf T_{i,:}$ to this equation can be described in terms of the adjugate of the $k\times k$ square matrix $\mathbf P=\begin{bmatrix} 0_{k\times 1}& \mathbf A_{:,\mathcal Z_i}\end{bmatrix}$.
    Recall that $\operatorname{adj}\mathbf P$ is the transpose of the cofactor matrix $\left[(-1)^{i+j}\det (\mathbf P_{[k]\backslash\{i\},[k]\backslash\{j\}})\right]_{i,j\in[k]}$ and satisfies
    $\operatorname{adj}(\mathbf P)\mathbf P=\det(\mathbf P)\mathbf I_{k\times k}$.
    Since $\mathbf P$ has an all zero column, we have $\det\mathbf P=0$, which implies $\operatorname{adj}(\mathbf P)\mathbf P=0$.
    Furthermore, due to the zero column in $\mathbf P$, the entries of $\operatorname{adj}\mathbf P$ are zero except the first row, whose entries are for $j\in[k]$,
    \begin{align*}
        (\operatorname{adj}\mathbf P)_{1,j} &= (-1)^{j+1}\det (\mathbf P_{[k]\backslash\{j\},[k]\backslash\{1\}})\\
        &= (-1)^{j+1}\det (\mathbf A_{[k]\backslash\{j\},\mathcal Z_i})\\
        &= \det\begin{bmatrix}\mathbf e_j & \mathbf A_{:,\mathcal Z_i}\end{bmatrix}=\mathbf{T}_{i,j}.
    \end{align*}
    Since $(\operatorname{adj}\mathbf P)_{1,:}\cdot\mathbf P=0$ and $(\operatorname{adj}\mathbf P)_{1,:}\cdot\mathbf A_{:,\mathcal Z_i}=0$, the row vector $\mathbf T_{i,:}=(\operatorname{adj}\mathbf P)_{1,:}$ satisfies $\mathbf T_{i,:}\cdot \mathbf A_{:,\mathcal Z_i}=0$.
\end{proof}

Furthermore, if $x_1,\dots,x_n$ are $\mathsf F$--linearly independent and the matrix $\mathbf T$ is invertible (i.e. $\det\mathbf T \neq 0)$,
then the code generated by $\mathbf G$ is an $[n,k]_{\mathsf E/\mathsf F}$ Gabidulin code since the row spaces of $\mathbf A$ and $\mathbf G=\mathbf T\cdot\mathbf A$ are identical.
In Theorem~\ref{thm:construction}, we give a lower bound on the probability of this construction giving an MRD code.

\begin{theorem}\label{thm:construction}
    If the conditions in Theorem~\ref{thm:main} are satisfied,
    then, the generator matrix $\mathbf G$ randomly constructed as described above will satisfy (\ref{eq:zeroconstraints}) and generate an $[n,k]_{\mathsf E/\mathsf F}$ Gabidulin code  with probability at least $1-\frac{n+k(k-1)}{|S|}$.
\end{theorem}

Since $\mathsf F$ is infinite, $S$ can be arbitrarily large. Therefore, the probability of constructing an MRD code can be arbitrarily close to~$1$.

Furthermore, if the $\mathcal Z_i$ do not satisfy (\ref{eq:ineqcond}), then by following the proof of Corollary \ref{corollary}, we can construct a rank metric code achieving the largest possible rank distance for the given support constraints.

\section{More on Cyclic Galois Extensions}\label{section:MoreOn}
Before moving to the proofs of the theorems, in this section, we will give some useful properties of the automorphisms in $\operatorname{Gal}(\mathsf E/\mathsf F)=\{\theta^0,\theta^1,\dots,\theta^{m-1}\}$.

\subsection{Linear independence of the elements in $\mathsf E$}
Lemma \ref{lemma:lindep} below lists some equivalent conditions to the $\mathsf F$--linear dependence of the elements of $\mathsf E$ in terms of the automorphisms in $\operatorname{Gal}(\mathsf E/\mathsf F)$.
The first two of these conditions can be also seen as a special case of \cite[Prop.~5]{augot2018generalized}, where the authors give equivalent rank metrics for the elements of $\mathsf E^n$, whereas Lemma \ref{lemma:lindep} only claims these rank metrics simultaneously declare rank deficiency (i.e. returns a rank less than~$n$) for a given element of $\mathsf E^n$.
It is worth noting,
as shown by Augot \textit{et al.} \cite{augot2018generalized}, that the assumption that the extension $\mathsf E/\mathsf F$ is cyclic plays an important role in Lemma~\ref{lemma:lindep}. This is since its proof relies on the fact that $\theta$ fixes \emph{only} the elements of $\mathsf F$ (i.e. for any $x\in\mathsf E$, $\theta(x)=x$ if and only if $x\in\mathsf F$),
which is the case for the cyclic extensions.

\begin{lemma}\label{lemma:lindep}
Let $n\leq m=\dim_{\mathsf F}\mathsf E$, $x_1,\dots,x_n\in\mathsf E$, and
\begin{align}\label{eq:matrixm}
    \mathbf M=\begin{bmatrix}
        \theta^0(x_1) & \theta^0(x_2) & \cdots & \theta^0(x_n)\\
        \theta^1(x_1) & \theta^1(x_2) & \cdots & \theta^1(x_n)\\
        \vdots & \vdots & & \vdots\\
        \theta^{m-1}(x_1) & \theta^{m-1}(x_2) & \cdots & \theta^{m-1}(x_n)
    \end{bmatrix}\in\mathsf E^{m\times n}
\end{align}
Then, the following are equivalent:
\begin{enumerate}[(i)]
    \item $x_1,\dots,x_n$ are $\mathsf F$--linearly dependent.
    \item The columns of $\mathbf M$ are $\mathsf E$--linearly dependent.
    \item The top $n\times n$ minor of $\mathbf M$
    is zero, i.e. $\det\mathbf M_{[n],[n]}=0$.
\end{enumerate}
\end{lemma}
\begin{proof} If $x_i=0$ for some~$i$ then the claim is trivial, and hence assume that $x_i\neq 0$ for every~$i$.

    $(ii)\implies(i)$: Let $\ell$ be the minimum number of columns of $\mathbf{M} $ that are $\mathsf E$--linearly dependent and w.l.o.g. assume that
    \begin{align*}
        \mathbf M_{:,\ell}=\sum_{i=1}^{\ell-1}\beta_i\mathbf M_{:,i}
    \end{align*}
    for some unique $\beta_1,\dots,\beta_{\ell-1}\in\mathsf E$, which implies that~$\theta^{j-1}(x_\ell)=\sum_{i=1}^{\ell-1}\beta_i\theta^{j-1}(x_i)$ for every~$j\in[m]$.
    Then, applying~$\theta$ to both sides gives
    $\theta^{j}(x_{\ell})=\sum_{i=1}^{\ell}\theta(\beta_i)\theta^{j}(x_i)$,
    which implies $\mathbf M_{:,\ell}=\sum_{i=1}^{\ell-1}\theta(\beta_i)\mathbf M_{:,i}$ as $\theta^m=\theta^0$.
    Since the $\beta_i$'s are unique it follows that $\theta(\beta_i)=\beta_i$, which implies $\beta_i\in\mathsf F$.
    Since $\theta^0(x)=x$, we have $x_{\ell}=\sum_{i=1}^{\ell}\beta_ix_i$ for $\beta_i\in\mathsf F$.
    
    $(iii)\implies (ii)$: If the top $n\times n$ minor of $\mathbf M$ is zero, then there exists $\ell\leq n$ such that the $\ell$'th row of $\mathbf M$ is in the $\mathsf E$--span of the first $\ell-1$ rows. By induction, it can be shown that for any $i\geq\ell$, the $i$'th row is in the span of the first $\ell-1$ rows. To see how, assume for some $\beta_1,\dots,\beta_{\ell-1}\in\mathsf E$, $\theta^{i-1}(x_j)=\sum_{t=1}^{\ell-1}\beta_t\theta^{t-1}(x_j)$ for all $j$. Then, by applying $\theta$ to both sides, it follows that the $(i+1)$'th row is a linear combination of the first $\ell$ rows; hence it is also in the span of the first $\ell-1$ rows. As a result, $\operatorname{rank}\mathbf M\leq\ell-1< n$, which implies $(ii)$.
    
    $(i)\implies (iii)$: Assume that $\sum_{i=1}^n\beta_ix_i=0$ for some $\beta_i\in\mathsf F$. Then, for any $j$, applying $\theta^j$ to both sides yields $\sum_{i=1}^n\beta_i\theta^j(x_i)=0$ since $\theta^j(\beta_i)=\beta_i$, which implies $(iii)$.
\end{proof}

\subsection{Schwartz--Zippel Lemma for automorphisms}
Recall the Schwartz--Zippel Lemma, which states that 
for a nonzero multivariate polynomial $f$ in $n$ variables over a field, a point uniformly chosen at random from $S^n$, where $S$ is a nonempty finite subset of this field, will be a root of $f$ with probability at most $\frac{\deg f}{|S|}$.
In this section, we will give an extension of Schwartz--Zippel Lemma for a special type of functions from $\mathsf E^n$ to $\mathsf E$.
More precisely, for a given multivariate polynomial $f$ over $\mathsf E$ in $mn$ variables (seen as an $m\times n$ matrix), we will consider the function $g(x_1,\dots,x_n)=f([\theta^{i-1}(x_j)]_{i\in[m],j\in[n]})$ and give a bound on the probability of a randomly chosen point being a zero of $g$.
Later, this will help us to derive the bound on the probability given in Theorem~\ref{thm:construction}.

\begin{lemma}\label{lemma:schwartz_zippel}
    Let $\{b_1,\dots,b_m\}$ be an $\mathsf F$--basis for $\mathsf E$.
    Let $f$ be a nonzero multivariate polynomial over $\mathsf E$ in $mn$ variables.
    Let $\mathbf M\in\mathsf E^{m\times n}$ be defined as in (\ref{eq:matrixm}) for $x_j=\sum_{i=1}^m\mathbf\Gamma_{ij}b_i$, where the $\mathbf\Gamma_{ij}$ are independently uniformly chosen at random from a finite nonempty subset $S\subset\mathsf F$.
    Then,
    \begin{align*}
        \mathbb P(f(\mathbf M)= 0) \leq \frac{\deg f}{|S|}.
    \end{align*}
\end{lemma}
\begin{proof}
    Define another polynomial $f'$ as
    $
        f'(\mathbf X) = f(\mathbf B\mathbf X)
    $
    in the variables $\mathbf X_{ij}$, $i\in[m],j\in[n]$,
    where $\mathbf B=[\theta^{i-1}(b_j)]_{i,j\in[m]}$ is an 
    ${m\times m}$ matrix defined as in (\ref{eq:matrixm}) for $b_1,\dots,b_m$.
    Since $\{b_1,\dots,b_m\}$ is an $\mathsf F$--basis, the $b_i$ are $\mathsf F$--linearly independent and
    by Lemma \ref{lemma:lindep}, $\mathbf B$ is invertible.
    Then, $f$ can be also written as
    $
        f(\mathbf X) = f'(\mathbf B^{-1}\mathbf X)
    $.
    Hence, $f'$ is also nonzero and $\deg f=\deg f'$.
    Furthermore,
    $
        f'(\mathbf\Gamma)
        =f(\mathbf B\mathbf\Gamma)
        =f(\mathbf M)
    $
    since 
    \begin{align*}
        \mathbf M_{ij}
        &= \theta^{i-1}(x_j)\\
        &= \theta^{i-1}\left(\textstyle\sum_{t=1}^mb_t\mathbf\Gamma_{tj}\right)\\
        &= \textstyle\sum_{t=1}^m\theta^{i-1}(b_t)\mathbf\Gamma_{tj}\\
        &= (\mathbf B\mathbf\Gamma)_{ij}
    \end{align*}
    where we use $\theta^{i-1}(\mathbf\Gamma_{tj})=\mathbf\Gamma_{tj}$ since $\mathbf\Gamma_{tj}\in\mathsf F$.
    Now, applying the Schwartz--Zippel Lemma to the polynomial $f'$ gives $\mathbb P(f'(\mathbf\Gamma)=0)\leq\frac{\deg f'}{|S|}$. Hence, $\mathbb P(f(\mathbf M)=0)\leq\frac{\deg f}{|S|}$.
\end{proof}

\section{Proofs of Theorem~\ref{thm:main} and Theorem~\ref{thm:construction}}\label{sec:proofs}
First of all, notice that it is sufficient to prove Theorem~\ref{thm:construction} since it implies Theorem~\ref{thm:main} when $S$ is chosen sufficiently large.
Assume $x_1,\dots,x_n$ are chosen as described in Theorem~\ref{thm:construction}.
We know that the code with the generator matrix $\mathbf T\cdot\mathbf A$, which satisfies (\ref{eq:zeroconstraints}) by Lemma \ref{lemma:general}, is an $[n,k]_{\mathsf E/\mathsf F}$ Gabidulin code if the $x_i$'s are $\mathsf F$--linearly independent and $\mathbf T$ is invertible.
Define $\mathbf M\in\mathsf E^{m\times n}$ as in Lemma \ref{lemma:lindep}, by which the $x_i$'s are $\mathsf F$--linearly independent iff $\det\mathbf M_{[n],:}\neq 0$.
Furthermore, since $\mathbf A = \mathbf M_{[k],:}$, we have that
\begin{align*}
    \mathbf T &= \left[\det\begin{bmatrix}\mathbf e_j&\mathbf A_{:,\mathcal Z_i}\end{bmatrix}\right]_{i,j\in[k]}
    = \left[\det\begin{bmatrix}\mathbf e_j&\mathbf M_{[k],\mathcal Z_i}\end{bmatrix}\right]_{i,j\in[k]}.
\end{align*}
Therefore, it is sufficient to show that
$\mathbb P(\det\mathbf T\cdot\det\mathbf M_{[n],:}\neq 0)\geq 1-\frac{n+k(k-1)}{|S|}$ or that $\mathbb P(\det\mathbf T\cdot\det\mathbf M_{[n],:}= 0)\leq \frac{n+k(k-1)}{|S|}$.

In order to show this, we will appeal to Lemma \ref{lemma:schwartz_zippel}.
Define the multivariate polynomial
\begin{align}\label{eq:polyf}
    f(\mathbf X) = \det\left(\left[\det\begin{bmatrix}\mathbf e_j&\mathbf X_{[k],\mathcal Z_i}\end{bmatrix}\right]_{i,j\in[k]}\right)\cdot \det\mathbf X_{[n],:}
\end{align}
for the variables $\mathbf X_{ij}$, $i\in[m],j\in[n]$ seen as an $m\times n$ matrix~$\mathbf X$.
Then, it suffices to show that $\mathbb P(f(\mathbf M)=0)\leq\frac{n+k(k-1)}{|S|}$.
Hence, by Lemma \ref{lemma:schwartz_zippel},
all we need to show is that $f$ is a nonzero polynomial with total degree at most $n+k(k-1)$.

To show the bound on the degree of~$f$, recall the Leibniz formula for the determinant of an $n\times n$ square matrix $\mathbf Z$, which is
$\det\mathbf Z = \sum_{\pi\in S_n}\operatorname{sgn}(\pi)\prod_{i=1}^n\mathbf Z_{\pi(i),i}$,
where $S_n$ is the permutation group of size $n$ and $\operatorname{sgn}(\pi)$ is the sign of the permutation $\pi$.
Thus, when the entries of $\mathbf Z$ are polynomials, we can write
\begin{align}
    \deg\det\mathbf Z\leq \sum_{j\in[n]}\max_{i\in[n]}\deg\mathbf Z_{i,j}.
\end{align}
Hence, $\deg\det\mathbf X_{[n],:}\leq n$ since each entry of $\mathbf X$ has degree one.
Furthermore,
$\deg\det\begin{bmatrix}\mathbf e_j&\mathbf X_{[k],\mathcal Z_i}\end{bmatrix}\leq k-1$; hence,
$\deg\det\left(\left[\det\begin{bmatrix}\mathbf e_j&\mathbf X_{[k],\mathcal Z_i}\end{bmatrix}\right]_{i,j\in[k]}\right)\leq k(k-1)$.
As a result, $\deg f\leq n+k(k-1)$.

To show that~$f$ is a nonzero polynomial, we will use the simplified GM--MDS conjecture of Dau \textit{et al.} \cite{dau2014existence}, which was proved in \cite{yildiz2019optimum} and \cite{lovett2018mds}.

\begin{lemma}[{Simplified GM--MDS conjecture \cite[Thm.~3]{yildiz2019optimum}\footnote{Compared to \cite[Thm.~3]{yildiz2019optimum}, in the statement of Lemma \ref{lemma:gmmds}, the variable $\alpha_j$ is replaced with $-\alpha_j$ and the matrix $\mathbf P$ is flipped about its vertical axis, which may only change the sign of the determinant.}}]
\label{lemma:gmmds}
    Let $\mathcal Z_1,\dots,\mathcal Z_k\subset[n]$ be subsets of size $k-1$.
    Then, they satisfy (\ref{eq:ineqcond}) if and only if the determinant of the $k\times k$ matrix
    \begin{align}\label{eq:matrixp}
        \mathbf P &= 
        \begin{bmatrix}
            \prod_{t\in\mathcal Z_1}(-\alpha_t)& \cdots & \sum_{t\in\mathcal Z_1}(-\alpha_t) & 1\\
            \prod_{t\in\mathcal Z_2}(-\alpha_t) & \cdots & \sum_{t\in\mathcal Z_2}(-\alpha_t) & 1\\
            \vdots & &\vdots & \vdots\\
            \prod_{t\in\mathcal Z_k}(-\alpha_t) & \cdots & \sum_{t\in\mathcal Z_k}(-\alpha_t) & 1
        \end{bmatrix}
    \end{align}
    with entries $\mathbf P_{ij} = \sum_{\mathcal S\subseteq \mathcal Z_i, |\mathcal S|=k-j}\prod_{t\in \mathcal S}(-\alpha_t)$
    is not the zero polynomial in the variables $\alpha_1,\dots,\alpha_n$.
\end{lemma}
Notice that the $i$'th row of $\mathbf P$ in (\ref{eq:matrixp}) consists of the coefficients of the polynomial
\begin{align}\label{equation:PisPoly}
\prod_{j\in\mathcal Z_i}(X-\alpha_j)=\sum_{j=1}^k\mathbf \mathbf \mathbf P_{ij}X^{j-1}
\end{align} 
in the variable~$X$.
We will also show that $\mathbf P$ can be written in the form of (\ref{eq:matrixt}). To see how, define the $m\times n$ Vandermonde matrix $\mathbf V=\left[\alpha_j^{i-1}\right]_{i\in[m],j\in[n]}$.
Fix $i\in[k]$ and consider the determinant of the $k\times k$ Vandermonde matrix
$\mathbf W=\begin{bmatrix}\mathbf v & \mathbf V_{[k],\mathcal Z_i}\end{bmatrix}$,
where $\mathbf v$ is a column vector whose $j$'th entry is $X^{j-1}$ for $j\in[k]$:
\begin{align*}
    \det\mathbf W = c_i\prod_{j\in\mathcal Z_i}(X-\alpha_j) \overset{\eqref{equation:PisPoly}}{=} c_i\sum_{j\in[k]}\mathbf P_{ij}X^{j-1}
\end{align*}
where $c_i=\prod_{j_1<j_2\in\mathcal Z_i}(\alpha_{j_1}-\alpha_{j_2})\neq 0$.
On the other hand, by the linearity of the determinant in the first column, we can write
\begin{align*}
    \det\mathbf W = \sum_{j\in[k]}\det\begin{bmatrix}\mathbf e_j & \mathbf V_{[k],\mathcal Z_i}\end{bmatrix} X^{j-1},
\end{align*}
since $\mathbf v=\sum_{j\in[k]}\mathbf e_jX^{j-1}$.
As a result, the entries of $\mathbf P$ satisfy
\begin{align}\label{eq:entriesP}
    c_i\mathbf P_{ij} = \det\begin{bmatrix}\mathbf e_j & \mathbf V_{[k],\mathcal Z_i}\end{bmatrix}
\end{align}

Now, let us evaluate $f$ in (\ref{eq:polyf}) at $\mathbf V$, which will give a multivariate polynomial in the variables $\alpha_j$:
\begin{align*}
    f(\mathbf V)
    &= \det\left(\left[\det\begin{bmatrix}\mathbf e_j&\mathbf V_{[k],\mathcal Z_i}\end{bmatrix}\right]_{i,j\in[k]}\right)\cdot \det\mathbf V_{[n],:}\\
    &\overset{\eqref{eq:entriesP}}{=} \det\left(\left[c_i\mathbf P_{ij}\right]_{i,j\in[k]}\right)\cdot \det\mathbf V_{[n],:}\\
    &= \det\mathbf P\cdot\left(\prod_{i\in[k]}c_i\right)\cdot\det\mathbf V_{[n],:}.
\end{align*}
By Lemma \ref{lemma:gmmds}, $\det \mathbf P$ is a nonzero polynomial. Furthermore, we have that $c_i\neq 0$ and $\det\mathbf V_{[n],:}=\prod_{j_1<j_2\in[n]}(\alpha_{j_1}-\alpha_{j_2})\neq 0$.
Hence, $f(\mathbf V)$ is not the zero polynomial in the variables $\alpha_j$.
Therefore, $f(\mathbf X)$ itself cannot be the zero polynomial in the variables $\mathbf X_{ij}$.
\hfill\qed

%\section{Conclusion}
%\td{Conclusion}

\bibliographystyle{IEEEtran}
\bibliography{refs}

\end{document}